\documentclass[showpacs,preprintnumbers,onecolumn]{revtex4}

\usepackage[centertags]{amsmath}
\usepackage{amsfonts}
\usepackage{amssymb}
\usepackage{amsthm}
\usepackage{newlfont}
\usepackage{graphicx}

\newtheorem{thm}{Theorem}[section]
\newtheorem{lem}[thm]{Lemma}

\newtheorem{defn}[thm]{Definition}

\newtheorem{prop}[thm]{Proposition}

\begin{document}
% Title portion
\title{Quantum Entanglement in Time}
\author{Marcin Nowakowski\footnote{Electronic address: mnowakowski@mif.pg.gda.pl}}
\affiliation{Faculty of Applied Physics and Mathematics,
~Gdansk University of Technology, 80-952 Gdansk, Poland}
\affiliation{National Quantum Information Center of Gdansk, Andersa 27, 81-824 Sopot, Poland}


\begin{abstract}
In this paper we present a concept of quantum entanglement in time in a context of entangled consistent histories. These considerations are supported by presentation of necessary tools closely related to those acting on a space of spatial multipartite quantum states. We show that in similarity to monogamy of quantum entanglement in space, quantum entanglement in time is also endowed with this property for a particular history. Basing on these observations, we discuss further bounding of temporal correlations and derive analytically the Tsirelson bound implied by entangled histories for the Leggett-Garg inequalities.
\end{abstract}

\maketitle
% Head 1
\section{INTRODUCTION}
Recent years have proved a great interest of quantum entanglement concept showing its broad application in quantum communication theory. Spatial quantum correlations and especially their non-locality became a central subject of quantum information theory and their applications to quantum computation, yet potential applications of temporal non-local correlations are poorly analyzed. The crucial issue relates to the very nature of time, thus, temporal correlations phenomenon is a subject of many open questions within the framework of modern quantum and relativistic theories.

Non-local nature of quantum correlations in space has been accepted as a consequence of violation of local realism, expressed in Bell's theorem \cite{Bell} and analyzed in many experiments \cite{Aspect, Freedman}. As an analogy for temporal correlations, the violation of macro-realism \cite{LGI2} and Legett-Garg inequalities \cite{LGI} seem to indicate non-local effects in time, and they are a subject of many experimental considerations \cite{EX1, EX2, EX3, EX4}.

In this paper we discuss a variation of the consistent histories approach \cite{RG1, RG2, RG3, RG4} with an introduced concept of entangled histories \cite{CJ1,CJ2} built on a tensor product of projective Hilbert spaces, that can be considered as a potential candidate of a mathematical structure representing quantum states entangled in time.

In particular, we focus on showing that entangled histories demonstrate monogamous properties reflecting the spatial phenomenon.

However, it is crucial to note that in this context many 'obvious' facts about structure and behavior of spatial correlations and tensor algebra of spatial quantum states cannot be easily transferred into the temporal domain as the tensor structure of temporal correlations is richer due to the binding evolution between instances of 'time' and the observation-measurement phenomenon that is also a subject of this paper.

The outline of this paper is as follows: in the first section, we present the key concepts of consistent histories approach \cite{RG3} and present some new concepts related to entangled histories \cite{WC1, WC2} which are substantial for analysis of monogamies and entanglement in time as such.

In the section related to monogamy of quantum entanglement in space, we recall the local realistic assumptions about physical reality and their implications articulated as Bell inequalities. We discuss also the concept of monogamy of spatial quantum entanglement.

In the section focused on quantum entanglement in time, we introduce partial trace on quantum histories and show that quantum entanglement in time is monogamous for a particular history. This section considers also this property from a perspective of the Feynman's path integral approach.

In the final section, the Tsirelson bound on quantum correlations in time for Legget-Garg inequalities is derived from the entangled histories.

\section{ENTANGLED HISTORIES}

The decoherent histories theory (or consistent histories theory) has a long tradition \cite{Hartle1, Hartle2, Hartle3, RG1,RG2, RG3,RG4,RG5, RG6, CJ1, CJ2} and is built on the ground of well known and broadly applied Feynman's path integrals theory \cite{Feynman} for calculation of probability amplitudes of quantum processes, especially in quantum field theory or quantum electrodynamics. It is presented also as a generalization of quantum mechanics applied to closed systems such as the universe as a whole and discussed as a necessary element of future quantum gravity theory \cite{Hartle1}.

For readers interested in deepening this matter, it might be useful to refer to the literature \cite{RG3, Hartle1, Hartle2, Hartle3, Hartle4}. In this section we focus on introduction to the concept of a consistent history and its recent modification, an entangled history \cite{WC1, WC2}. We present also a proposal of the temporal partial trace operator \cite{MNPH5} acting on $\mathcal{C}^{*}$-Algebra of history operators as a tool necessary to achieve reduced histories, in similarity to the partial trace operator acting on a multipartite quantum state.

For the sake of the concept of a consistent history, it is substantial to note that for an evolving system (e.g. a non-relativistic particle being in an initial state $|\psi_{0}\rangle$ which evolution is governed by the Hamiltonian $H$), we can ask questions about the states of the system at different
times $t_{1}<t_{2}<...<t_{n}$. It could be performed during the repeating measuring process where a question at time $t_{x}$ could be represented naturally by a projector $P_{x}$. The alternatives at a given time $t_{x}$ form an exhaustive orthogonal set of projectors $\{P_{x}^{\alpha_{x}}\}$  where:
\begin{eqnarray}
&&\sum_{\alpha_{x}}P_{x}^{\alpha_{x}}=\mathbb{I} \\
&&P_{x}^{\alpha_{x}}P_{x}^{\tilde{\alpha}_{x}}=\delta_{\alpha_{x}\tilde{\alpha}_{x}} P_{x}^{\alpha_{x}}
\end{eqnarray}

Therefore, the alternative histories could be represented by the sets of alternative operators $\{P_{1}^{\alpha_{1}}\}$, $\{P_{2}^{\alpha_{2}}\}$,\ldots, $\{P_{n}^{\alpha_{n}}\}$ at different times $t_{1}<t_{2}<...<t_{n}$. A particular history is then represented as a tensor product $Proj(\mathcal{H})\ni |H)=P_{n}^{\alpha_n}\odot P_{n-1}^{\alpha_{n-1}}\odot\ldots\odot P_{1}^{\alpha_1}$. This could be perceived that the system had a property $P_{i}^{\alpha_{i}}$ at time $t_{i}$ \cite{RG3}.

We could interpret that during this process we project the global state of the system onto the n-fold tensor product $\bigodot_{i=1}^{n}P_{i}^{\alpha_{i}}$ achieving a consistent wave function which can be used to deduce probabilities of the events \cite{MNPH5} in accordance with the Born rule.

The fundamental tool introduced in the consistent history framework which connects different times is the bridging operator \cite{RG1} $\mathcal{B}(t_{2},t_{1})$. It is a counterpart of an unitary evolution operation having the following properties:
\begin{eqnarray}
\mathcal{B}(t_{2},t_{1})^{\dag}&=&\mathcal{B}(t_{1},t_{2}) \\
\mathcal{B}(t_{3},t_{2})\mathcal{B}(t_{2},t_{1})&=&\mathcal{B}(t_{3},t_{1})
\end{eqnarray}
and can be represented for a unitary quantum evolution as $\mathcal{B}(t_{2},t_{1})=\exp(-iH(t_{2}-t_{1}))$ (with the evolution governed be a Hamiltonian $H$).

Since we assumed for a given time that $\sum_{\alpha_{x}}P_{x}^{\alpha_{x}}=\mathbb{I}$, for the sample space of consistent histories $|H^{\overline{\alpha}})=P_{n}^{\alpha_{n}}\odot P_{n-1}^{\alpha_{n-1}}\odot\ldots\odot P_{1}^{\alpha_{1}}\odot P_{0}^{\alpha_{0}}$ ($\overline{\alpha}=(\alpha_{n}, \alpha_{n-1},\ldots, \alpha_{0})$) there holds $\sum_{\overline{\alpha}}|H^{\overline{\alpha}})=\mathbb{I}$.

Further, the consistent histories formalism introduces the chain operator $K(|H^{\overline{\alpha}}))$ which can be directly associated with a time propagator of a given quantum process:
\begin{equation}
K(|H^{\overline{\alpha}}))=P_{n}^{\alpha_n}\mathcal{B}(t_{n},t_{n-1}) P_{n-1}^{\alpha_{n-1}}\ldots\mathcal{B}(t_{2},t_{1})P_{1}^{\alpha_1}\mathcal{B}(t_{1},t_{0})P_{0}^{\alpha_0}
\end{equation}

Equipped with this operator, one can associate a history $|H^{\alpha})$ with its weight:
\begin{equation}
W(|H^{\alpha}))=TrK(|H^{\alpha}))^{\dagger}K(|H^{\alpha}))
\end{equation}
being by Born rule a counterpart of relative probability and can be interpreted as a probability of a history realization.

As an example, suppose that the system is in a state $|\psi_{0}\rangle\in\mathcal{H}$ at time $t_{0}$ and evolves to time $t_{2}$ under the bridging operator $\mathcal{B}(t_{1},t_{0})$, then applying the Born rule one can determine the probability that the system at time $t_{1}$ has a property $P_{t_{1}}$:
\begin{eqnarray}
Pr(P_{t_{1}},t_{1})&=&\|P_{t_{1}}\mathcal{B}(t_{1},t_{0})|\psi_{0}\rangle\|^2 \\
&=& \langle\psi_{0} |\mathcal{B}^{\dag}(t_{1},t_{0}) P_{t_{1}} \mathcal{B}(t_{1},t_{0})|\psi_{0}\rangle\\
&=& Tr(\mathcal{B}^{\dag}(t_{1},t_{0}) P_{t_{1}} \mathcal{B}(t_{1},t_{0})[\psi_{0}])\\
\end{eqnarray}
where $[\psi_{0}]=|\psi_{0}\rangle\langle \psi_{0}|$ as discussed further.

The set of histories is coarse-grained as the alternatives are defined for chosen times, yet not for every possible time \cite{Hartle1, Hartle2}. It means that the set of potential histories is partitioned into the set of mutually exclusive classes called coarse-grained histories, those which are observable during the process of measurements. Coarse graining of measurements is a natural feature of "standard" quantum mechanics. The consistent histories theory describes also fine-grained histories and relations between the sets of coarse-grained and fine-grained histories, however, this is not a subject of this presentation and it does not change generality of the following conclusions.

Recent years show also an extensive discussion about a subject of the so-called consistency or decoherence of allowed histories \cite{Hartle1, Hartle2, Hartle3} which is directly related to the degree of interference between pairs of histories in the set of histories. The consistent histories framework assumes that the family \footnote{The family of consistent histories is such a set of histories $\mathcal{F}=\{|H^{\overline{\alpha}})\}_{\overline{\alpha}=(\alpha_{n}, \alpha_{n-1},\ldots, \alpha_{0})}$ that $\sum_{\overline{\alpha}}|H^{\overline{\alpha}})=\mathbb{I}$ and any pair of histories from the set meets the consistency condition.}
of histories is consistent, i.e. one can associate with a union of histories a weight equal to the sum of weights
associated with particular histories included in the union \cite{WC1, WC2}. This implies the following \textit{consistency condition}:
\begin{equation}
 \left\lbrace
           \begin{array}{l}
(H^{\alpha}|H^{\beta})\equiv TrK(|H^{\alpha}))^{\dagger}K(|H^{\beta}))=0 \; for\; \alpha\neq\beta \\
 (H^{\alpha}|H^{\beta})=0 \; or \; 1 \\
 \sum_{\alpha}c_{\alpha}|H^{\alpha})=I\; for\; c_{\alpha} \in \mathbb{C} \\
\end{array}
         \right.
\end{equation}

There are different conditions for the so-called decoherence functional $TrK(|H^{\alpha}))^{\dagger}K(|H^{\beta}))$ discussed, including the weaker condition that $TrK(|H^{\alpha}))^{\dagger}K(|H^{\beta}))\approx \delta_{\alpha\beta} P(\alpha)$ (medium decoherence and $P(\alpha)$ standing for probability of a history $|H^{\alpha})$) or the linear positivity condition by Goldstein and Page \cite{Goldstein}, however, as observed by F. Wilczek \cite{WC1, WC2}, it is unclear at this moment if the variants are significant.

It is helpful to assume normalization of histories with non-zero weight which enables normalization of probability distributions for history events, i.e.: $|\widetilde{H})=\frac{|H)}{\sqrt{(H|H)}}$ \cite{WC1, WC2, MNPH5}.

If the observed system starts its potential history in a pure state
$P_{t_{0}}=|\Psi_{0}\rangle\langle\Psi_{0}|$, then a consistent set of its histories create a tree-like structure (Fig. 1). Further, the consistency condition implies that the tree branches are mutually orthogonal.

The consistent history framework does not consider non-locality in space or time as such \cite{RG7}, however, since the space of histories spans the complex vector space, we can consider
complex combinations of history vectors, i.e. any history can be represented as \cite{WC1}:
\begin{equation}
 |\Psi)=\sum_{i}\alpha_{i}|H^{i})
\end{equation}
where $\alpha_{i} \in \mathbb{C}$ and $\mathcal{F} \ni |H^{i})$ represents a consistent family of histories which is actually a complex extension of the consistent histories framework.

Having defined above, the histories space can be also equipped with the inner semi-definite product \cite{RG1} between any two histories $|\Psi)$ and $|\Phi)$:
\begin{equation}
(\Psi|\Phi)=Tr [K(|\Psi))^{\dagger}K(|\Phi))].
\end{equation}

It is fundamental to note that a history $|H^{\alpha})$ can be consistent or inconsistent (physically not realizable) basing on the associated evolution $\mathcal{B}$ of the system \cite{RG3} as its consistency is verified by means
of the aforementioned inner product engaging bridging operators. Thus, a temporal history is always associated with evolution and for completeness, there should be considered a pair consisting of a family of histories and the bridging operators  $\{\mathcal{F}, T\}$. Whenever we analyze features of a spatial pure quantum state, it is assumed that all necessary knowledge is hidden in the vector $|\psi\rangle$ so actually we analyze only one-element history objects $[\psi]=|\psi\rangle\langle\psi|$
from a perspective of a temporal local frame.

\begin{figure}[h]
  \centerline{\includegraphics[width=10cm]{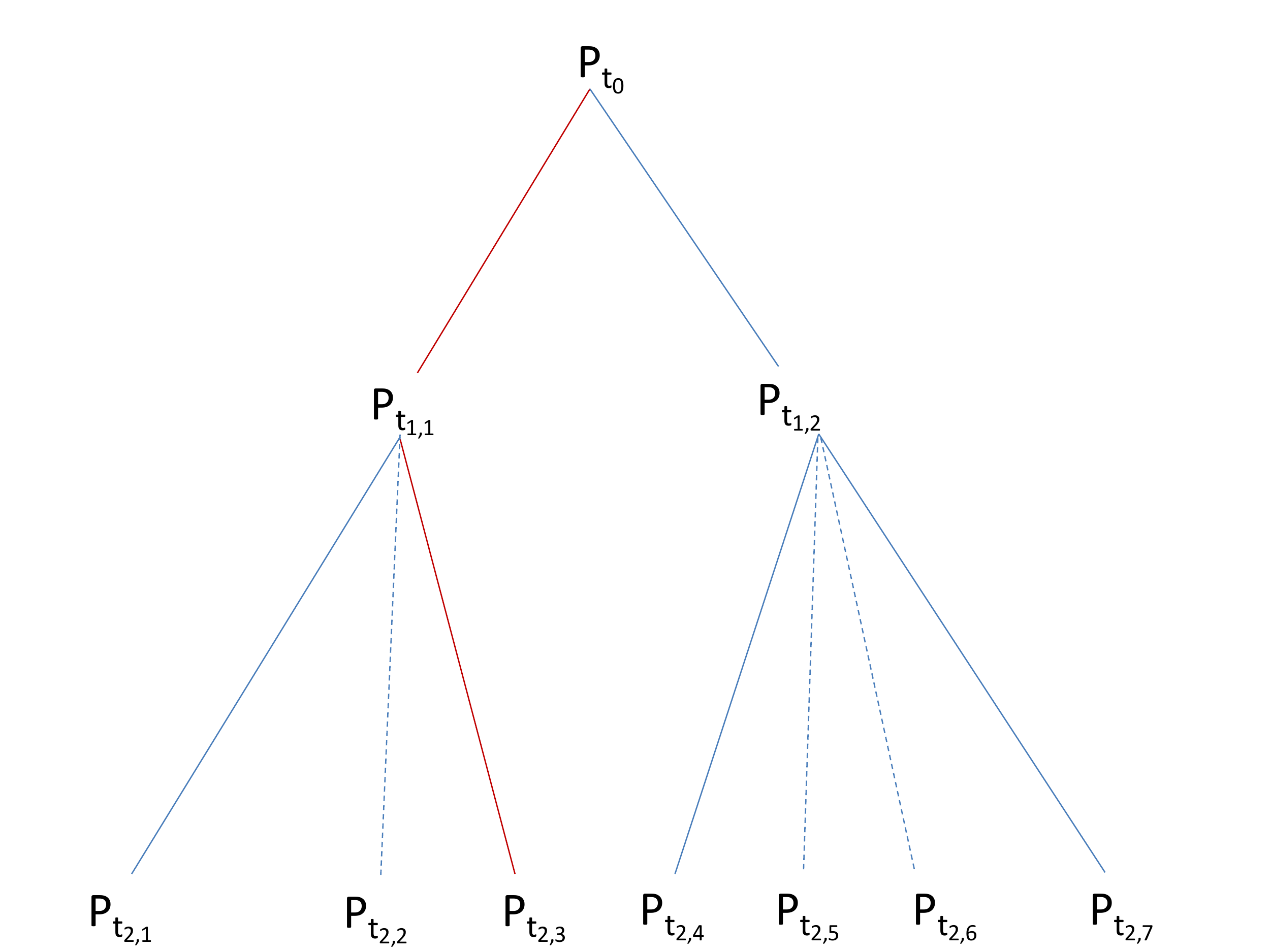}}
  \caption{If the observed evolution is initiated in a state $[P_{t_{0}}]=|\Psi_{0}\rangle\langle\Psi_{0}|$, then the history family can be represented as a tree-like structure where each branch represents a potential history. The branches are mutually orthogonal due to the consistency condition. The exemplary red branch represents history $|H)=P_{t_{2,3}}\odot P_{t_{1,1}}\odot P_{t_{0}}$.}
\end{figure}

\section{MONOGAMY OF QUANTUM ENTANGLEMENT IN SPACE}

Quantum entanglement is a phenomenon which does not have any reflection in classical world and as such is a manifestation of the so-called non-locality of quantum correlations. The roots of studies in this matter reach the year 1935 when the famous paper by Einstein, Podolsky and Rosen \cite{EPR} was published on nowadays' called (EPR) pairs, being in a maximally entangled state $|\Psi^{-}\rangle=\frac{1}{\sqrt{2}}(|01\rangle-|10\rangle)$ shared between two particles. In such a case none of the subsystems can have assigned a pure state.

In particular, many entangled states violate \textit{local realism} and as a consequence, Bell inequalities \cite{Bell}. Local realism has roots in classical world-view where for particular measurement of physical quantities, one believes that the measured physical quantities characterizing a physical object have a priori set values independent of the observers (realism) and for a bipartite setup, the measurement on one site does not influence the results of the other site's measurements (locality):

\textit{Realism.} The physical quantities being a subject of the measurements have definite real values which exist independent of the observation act.

\textit{Locality.} The results of measurements performed by Alice do no influence the results of measurements performed by Bob.

It is worth mentioning that the experiment is arranged in such a way that for two parties Alice and Bob, their experiments are causally disconnected. Thus, the measurement performed by Alice cannot influence the measurements done by Bob due to the light speed limit imposed by the special relativity theory.

To analyze correlations between results achieved in the experiment performed by Alice and Bob, imagine that they share a bipartite physical system consisting of two \textit{spatially} separated sub-systems that could interact in the past and which will be a subject of local measurements in distant laboratories belonging to Alice and Bob respectively (i.e.\textit{a distant lab paradigm}). Now, we can assign conditional probabilities to the measurement results $P(a,b|x,y)$ where $x$ and $y$ stand for measurement settings set locally by Alice and Bob respectively, and $a$ and $b$ for the measurement outcomes. Note that the measurement outcomes can be naturally inter-dependent, i.e. $P(a,b|x,y)\neq P(a|x)P(b|y)$ - the dependency can be created by \textit{a local hidden variable} $\lambda\in\Lambda$ that the experimenters are not aware of.
The hidden variables are a building block behind Bell inequalities and as such represent a hidden knowledge that cannot be possessed during the measurement process but can influence the measurement results and correlate them. The hidden variable is obviously also pre-defined in accordance with the local realism.

Since the local measurements' results are dependant only on x-settings and $\lambda$-variable for Alice and respectively on y-settings and $\lambda$-variable for Bob in local hidden variables (LHV) model, and moreover, we assume locality, then:
\begin{equation}
P(a,b|x,y,\lambda)=P(a|x,\lambda)P(b|y\lambda)
\end{equation}

For discrete distribution of $\lambda$ on $\Lambda$-space, after many measurement series we obtain (it reflects a random character of $\lambda$ in many measurements repeated on the system):
\begin{equation}
P(a,b|x,y)=\sum_{\lambda\in\Lambda}p(\lambda) P(a|x,\lambda)P(b|y,\lambda)
\end{equation}
For continuous distribution of $\lambda$ on $\Lambda$-space, we get a local hidden variable model:
\begin{equation}
P(a,b|x,y)=\int_{\Lambda} p(\lambda) P(a|x,\lambda)P(b|y,\lambda) d\lambda
\end{equation}

As a consequence of the local realism, every linear combination over such probabilities, meeting local realism conditions, builds the famous Bell inequalities for bipartite setup $\mathbf{B}(A,B)$ of an experiment performed by Alice and Bob. The Bell inequalities can be represented as a linear combination of conditional probabilities $P(a,b|x,y)$ (R is a local realistic bound - a real number):
\begin{equation}
\mathbf{B}(A,B)\equiv \sum_{xy}\sum_{ab}\alpha(a,b,x,y)P(a,b|x,y)\leq R
\end{equation}
and $\alpha(a,b,x,y)\geq 0$ parameters characterize the specific Bell inequality (since any Bell operator is a linear operator over conditional probabilities, one can always re-scale some initially negative $\alpha(a,b,x,y)$ - parameters so that in the re-scaled inequality $\alpha(a,b,x,y)\geq 0$).
These inequalities have to be met by all classical correlations with the aforementioned probability distributions $P(a,b|x,y)$ built for LHV models.

As observed, many entangled states violate the Bell inequalities, and e.g. the state $|\Psi_{-}\rangle=\frac{1}{\sqrt{2}}(|01\rangle-|10\rangle)$ - a maximally entangled state on $\mathbb{C}^{2}\otimes\mathbb{C}^{2}$ - violates the famous CHSH \cite{CHSHClauser} inequality maximally, saturating the Tsirelson bound \cite{Tsirelson}. In general, there exists a Bell inequality for any non-separable state which is violated by this state - this is an implication of the Hahn-Banach theorem for convex sets.

One of the fundamental questions related to quantum entanglement, as a resource shared between two parties Alice and Bob, is whether the correlations could be shared between more parties. The questions is fundamental not only due to applications in quantum computation or quantum cryptography but also due to the very nature of processing information between physical systems at different levels of complexity. It finds out that shareability of quantum correlations is bounded and it has its roots in $\textit{monogamy of quantum entanglement}$.

One can refer to a broadly used explanation \cite{Wootters} for spatial monogamy of entanglement
between parties ABC. It states that A cannot be simultaneously fully entangled with B and C since then $AB$ would be entangled with C having a mixed density matrix that contradicts purity of the singled state shared between A and B. It is expressed in Coffman-Kundu-Wootters (CKW) \cite{Wootters} monogamy inequality for three-qubit system in a state $\rho_{ABC}$:
\begin{equation}
C^2(\rho_{A|BC}) \leq C^2(\rho_{AB}) + C^2(\rho_{AC})
\end{equation}
where $C(\cdot)$ stands for the concurrence between the parties (e.g. $C^2(\rho_{A|BC})$ between A and BC subsystems). $C(\rho_{AB})$ is an entanglement monotone, and is defined as the averaged concurrence of an ensemble of pure states $\{p_{i}, |\Psi^{AB}_{i}\rangle\}$ corresponding to $\rho_{AB}$ minimized over all pure decompositions of $\rho_{AB}=\sum_{i}p_{i}|\Psi^{AB}_{i}\rangle\langle\Psi^{AB}_{i}|$ \cite{Wootters}:
\begin{equation}
C(\rho_{AB})=inf\sum_{i}p_{i}C(|\Psi^{AB}_{i}\rangle)
\end{equation}
and respectively for all other states. Concurrence of a pure state is  $C(|\Psi^{AB}\rangle)=\sqrt{2[1-Tr(\rho^{2}_{A})]}$ and $\rho_{A}=Tr_{B}|\Psi^{AB}\rangle\langle\Psi^{AB}|$.

If a bipartite state is in a state $\rho_{AB}=|\Psi^{-}\rangle\langle\Psi^{-}|$, then clearly the only possible tripartite extensions are of the form $\rho_{ABE}=\rho_{AB}\otimes\rho$, i.e. no symmetric extension of $|\Psi^{-}\rangle$ exists. That is also an immediate implication of the Schmidt decomposition for any purification of $\rho_{ABE}$ to a state $\Psi_{ABEE'}$ which has to be decomposed to a factorized state $\Psi_{ABEE'}=|\Psi^{-}\rangle\otimes|\Phi_{EE'}\rangle$ if for its reduction $AB$ one wants to get $\rho_{AB}=|\Psi^{-}\rangle\langle\Psi^{-}|$. Thus, we get at least two proofs of spatial monogamy of entanglement, one based on entanglement measures and one based on purely geometrical considerations.

%The concept of symmetric extendibility is directly related to monogamy of quantum entanglement and that phenomenon was a building block for initiation of broad %studies of symmetric extendibility applications. If a bipartite state is in a singlet state $\rho_{AB}=|\Psi^{+}\rangle\langle\Psi^{+}|$, then clearly the only %possible tripartite extensions are of the form $\rho_{ABE}=\rho_{AB}\otimes\rho$ as aforementioned and no symmetric extension of $|\Psi^{+}\rangle$ exists.

%\textit{Symmetric extendibility} \cite{D1, D2, T1, MN1} of a given bipartite state
%$\rho_{AB}\in \mathcal{B}(\mathcal{H}_A\otimes\mathcal{H}_B)$ (the Banach space of bounded operators) denotes that there
%exists a tripartite state $\rho_{ABE}\in
%\mathcal{B}(\mathcal{H}_A\otimes\mathcal{H}_B\otimes\mathcal{H}_B)$ invariant due to
%permutation of B and E part.
%\begin{defn} (Symmetric extendible extension)
%A state $\rho_{AB}\in \mathcal{B}(\mathcal{H}_A\otimes\mathcal{H}_B)$ is symmetrically extendible if there exists such a state $\rho_{ABE}\in
%\mathcal{B}(\mathcal{H}_A\otimes\mathcal{H}_B\otimes\mathcal{H}_B)$ ($\mathcal{H}_B=\mathcal{H}_E$) so that for permutation:
%\begin{equation}
%  P=\sum_{ijk}|ijk\rangle\langle ikj|
%\end{equation}
%there holds $P\rho_{ABE}P^{\dagger}=\rho_{ABE}$ and $Tr_{E}\rho_{ABE}=\rho_{AB}=\rho_{AE}$.
%\end{defn}

\section{QUANTUM ENTANGLEMENT IN TIME}
We consider in this section a concept of entanglement in time basing on the entangled consistent histories framework.

Since the algebra of histories with $\odot$ operation is a form of tensor algebra, it inherits the properties of a standard tensor algebra with $\otimes$ operation and all mathematical questions valid for vectors representing spatial correlations are mathematically valid for temporal correlations although not necessarily having similar physical interpretation \cite{MNPH5}.

Quantum entanglement for spatial correlations shared between two parties A and B, say Alice and Bob, denotes that the state $\rho_{AB}\in\mathcal{B}(\mathcal{H}_{A}\otimes\mathcal{H}_{B})$ of their bipartite system cannot be represented as a convex combination $\rho_{AB}=\sum_{i}p_{i}\rho_{A}^{i}\otimes\rho_{B}^{i}$ (which represents a separable state). The maximally entangled state of a bipartite system shared between Alice and Bob in space is represented as $|\Psi_{AB}\rangle=\frac{1}{\sqrt{d}}\sum_i |ii\rangle$. For the sake of spatial quantum entanglement, it is crucial to define the reductions of multipartite states and their extensions \cite{MNPH4}. To find a reduced state $\rho_{A}$ of a local state possessed e.g. by Alice, it is necessary to trace out Bob's system from the bipartite state $\rho_{AB}$ which is performed by the partial trace operation:
\begin{equation}
\rho_{A}=Tr_{B}\rho_{AB}=\sum_{i}\langle i_B|\rho_{AB}|i_B \rangle
\end{equation}
where the operation can be performed in a computational basis $|i_{B}\rangle$ of B-subsystem.

We will conduct further a similar reasoning for reductions of entangled histories, defining an operation of a partial trace over chosen times of a particular history state.

Now, it is substantial to note that any history $|Y)=F_{n}\odot\ldots\odot F_{0}$ can be extended to $I\odot Y$ as identity I represents a property that is always true and does not introduce additional
knowledge about the system.
Conversely, if one considers reduction of a history to smaller number of time frames, then information about the past and future of the reduced history is lost. Let us consider the potential history of the
physical system $|Y_{t_{n}\ldots t_{0}})=F_{n}\odot F_{n-1}\odot\ldots\odot F_{2}\odot F_{1}\odot F_{0}$ on times $\{t_{n}\ldots t_{0}\}$, then at time $t_{1}$ the reduced history is $|Y_{t_{1}})=F_{1}$. That was the trivial case of factorizable history, in analogy to factorizable quantum states in space, e.g. for $|\phi_{ABE}\rangle=|\phi_{A}\rangle\otimes|\phi_{B}\rangle\otimes|\phi_{E}\rangle$, the reduction over E results in $|\phi_{AB}\rangle=|\phi_{A}\rangle\otimes|\phi_{B}\rangle$. To find reductions over complex superpositions of histories, it is necessary to define a partial trace operator over a history.

%Further, looking at the history $Y_{t_{n}\ldots t_{0}}$ one can associate with %$Y_{t_{1}}$ two bridging operators $\mathcal{B}(t_{2},t_{1})^{\dag}$ and %$\mathcal{B}(t_{1},t_{0})$
%by means of which we can calculate two propagators taking the history from the %future and past events to $F_{1}$ where a system evolves through the potential %paths consistent with this history.

%....that means there are infinitely many potential probability amplitudes linked with the paths that lead from the future and the past of the analyzed history.

In analogy to partial trace for spatial quantum states, we introduced in \cite{MNPH5} a partial trace operation on a history state in accordance with general rules of calculating partial traces on tensor algebras:\\
\begin{defn}
For a history $|Y_{t_{n}\ldots t_{0}})$ acting on a space $\mathcal{H}=\mathcal{H}_{t_{n}}\otimes\dots\otimes\mathcal{H}_{t_{0}}$, a partial trace over times $\{t_{j}\ldots t_{i+1} t_{i}\}$ $(j\geq i)$ is:
\begin{equation}\label{PartialTrace}
Tr_{t_{j}\ldots t_{i+1}t_{i}} |Y_{t_{n}\ldots t_{0}})(Y_{t_{n}\ldots t_{0}}|=\sum_{k=1}^{\dim\mathcal{F}} (e_{k}|Y_{t_{n}\ldots t_{0}})(Y_{t_{n}\ldots t_{0}}|e_{k})
\end{equation}
where $\mathcal{F}=\{|e_{k})\}$ creates an orthonormal consistent family of histories on times $\{t_{j}\ldots t_{i+1} t_{i}\}$ and the strong consistency condition for partial histories holds for base histories, i.e. $(e_{i}|e_{j})=Tr[K(|e_{i}))^{\dag}K(|e_{j}))]=\delta_{ij}$.
\end{defn}

%\begin{figure}
%\includegraphics[width=7cm, height=5cm]{CHistories.pdf}
% \caption[Fig1.]{If the observed evolution is initiated in a state $[P_{t_{0}}]=|\Psi_{0}\rangle\langle\Psi_{0}|$, then the history family can be represented as a %tree-like structure where each brach represents a potential history. The branches are mutually orthogonal due to the consistency condition. The exemplary red %branch represents history $H=P_{t_{2,3}}\odot P_{t_{1,1}}\odot P_{t_{0}}$.}
% \label{Fig 1}
%\end{figure}

We propose further a general form of a maximally entangled history, in similarity to the maximally entangled state of a bipartite system in space, $|\Psi_{+}\rangle=\frac{1}{\sqrt{N}}\sum_{i=1}^{N}|i\rangle\otimes|i\rangle,\; 2\leq N<\infty$:\\
\begin{prop}
A history state 'maximally entangled' in time is represented by:
\begin{equation}\label{maxent}
|\Psi)=\frac{1}{\sqrt{N}}\sum_{i=1}^{N}|e_{i})\odot|e_{i}),\; 2\leq N<\infty
\end{equation}
with a trivial bridging operator $I$ and  $\{|e_{i})\}$ creating an orthonormal consistent histories family.
\end{prop}
It is important to note that one can always employ such a bridging operator that $|\Psi)$ could become intrinsically inconsistent which means it would be dynamically impossible \cite{RG3}, thus, an identity bridging operator is associated with the above state.
%Further, one could also introduce $\tau GHZ$ and $\tau W$ states substantial for studies of multipartite correlations and their applications (e.g. for secret key %generation, quantum algorithms or spin networks) in analogy to spatial $|GHZ\rangle$ and $|W\rangle$ states with trivial bridging operators:
%\begin{equation}
% \left\lbrace
%           \begin{array}{l}
%|\tau GHZ)=  \frac{1}{\sqrt{2}}( |e_{0})^{\odot N}+|e_{1})^{\odot N})\\
%|\tau W)= \frac{1}{\sqrt{N}}( |e_{1})\odot|e_{0})\odot\cdots \odot|e_{0}) +\\|e_{0})\odot|e_{1})\odot\cdots \odot|e_{0})+\cdots+|e_{0})\odot|e_{0})\odot\cdots %\odot|e_{1}))\\
%\end{array}
%         \right.
%\end{equation}

%\section{MONOGAMY OF QUANTUM ENTANGLEMENT IN TIME}

It is worth mentioning that recently \cite{WC1, WC2} the concept of Bell-like tests have been proposed for experimental analysis of entangled histories.
We further consider the Mach-Zehnder interferometer (Fig. 2,  $H=\frac{1}{\sqrt{2}} \left[
  \begin{array}{cc}
  1 & 1 \\
  1 & -1 \\
\end{array}\right]$) to discuss the matter of monogamy of quantum entanglement in time \cite{MNPH5}.

In the following let us consider an intrinsically consistent history on times $\{t_{3}, t_{2}, t_{1}, t_{0}\}$:
\begin{equation}
|\Lambda)=\alpha([\phi_{3,1}]\odot I_{t_{2}}\odot [\phi_{1,1}]+[\phi_{3,2}]\odot I_{t_{2}}\odot [\phi_{1,2}])\odot [\phi_{0}]
\end{equation}
where $\alpha$ stands for the normalization factor, $[\phi_{i,j}]=|\phi_{i,j}\rangle\langle\phi_{i,j}|$ and potentiality of the history means that one can construct a history observable $\widehat{\Lambda}=|\Lambda)(\Lambda|$.
Now, after tracing out the time $t_2$, one gets the reduced history on times $t_{1}$ and $t_{3}$:
\begin{equation}
|\Lambda_{1})=\tilde{\alpha}([\phi_{3,1}]\odot [\phi_{1,1}]+[\phi_{3,2}]\odot [\phi_{1,2}])
\end{equation}
which displays entanglement in time apparently.
Noticeably, we have to show that to be in agreement with the partial trace definition and Feynman propagators' formalism \cite{Feynman}, the history $|\Lambda_{1})$ cannot be extracted from
the following $|\tau GHZ)$-like state $|\Psi)$, i.e. $|\Lambda_{1})(\Lambda_{1}|\neq Tr_{t_2}|\Psi)(\Psi|$ \cite{MNPH5}.

We stress that the history state $|\Psi)$ is also allowed in the setup of the aforementioned interferometer (Fig. 2) as a potential history:
\begin{equation}
|\Psi)=\gamma([\phi_{3,1}]\odot [\phi_{2,1}]\odot [\phi_{1,1}]+[\phi_{3,2}]\odot [\phi_{2,2}]\odot [\phi_{1,2}])
\end{equation}

We observe that the reduced history $[\phi_{3,1}]\odot [\phi_{1,1}]$ is correlated with $[\phi_{2,1}]$ and not with $[\phi_{2,2}]$. Thus, we cannot simply add the histories $[\phi_{3,1}]\odot [\phi_{1,1}]+[\phi_{3,2}]\odot [\phi_{1,2}]$ as a reduction of $|\Psi)$ over time $t_2$.
It would imply decorrelation with the next instance of the history in such a case, i.e. it could be always expanded to a history e.g. $[\phi_{t_{x}}]\odot([\phi_{3,1}]\odot [\phi_{1,1}]+[\phi_{3,2}]\odot [\phi_{1,2}])$. This result is
in agreement with the Feynman's addition rule for probability amplitudes since this scenario would mean e.g. existence of detectors in the consecutive step performing measurements of the light states.

Imagine that for a maximally entangled history (\ref{maxent}) $\rho_{t_{1}t_{2}}=|\Psi)(\Psi|$ on times $\{t_1, t_2\}$ there exists a purification to a history state $|H_{t_{1}t_{2}t_{3}t_{4}})$ then in accordance with the partial trace definition (\ref{PartialTrace}), the maximally entangled history would have to be a reduction of $|H_{t_{1}t_{2}t_{3}t_{4}})(H_{t_{1}t_{2}t_{3}t_{4}}|$,  i.e. $|\Psi)(\Psi|=\sum_{i}(e_{i}|\odot I_{t_{1}t_{2}} |H_{t_{1}t_{2}t_{3}t_{4}})(H_{t_{1}t_{2}t_{3}t_{4}}| I_{t_{1}t_{2}} \odot |e_i)$ for some consistent history family $\mathcal{F}=\{|e_i)\}$ on times $\{t_{3},t_{4}\}$ but due to the consistency condition, one gets $|\Psi)$ only if $|H_{t_{1}t_{2}t_{3}t_{4}})=|H_{t_{3}t_{4}})\odot |\Psi)$ (for some history $|H_{t_{3}t_{4}})=\sum_{i}\gamma_{i}|e_{i})$, where $\gamma_{i}$ are complex numbers), otherwise the reduction would be a mixture of some consistent histories from the family. One can also observe immediately that $|\Psi)(\Psi|=\sum_{i}(e_{i}|\odot I_{t_{1}t_{2}} |H_{t_{1}t_{2}t_{3}t_{4}})(H_{t_{1}t_{2}t_{3}t_{4}}| I_{t_{1}t_{2}} \odot |e_i)$ implies that for any family base vector $|e_{i})$,  $|\Psi)=\gamma_i (e_{i}|\odot I_{t_{1}t_{2}} |H_{t_{1}t_{2}t_{3}t_{4}})$ (for some complex amplitude $\gamma_i$) and $(e_i|\odot(\Psi|\Psi')\odot |e_i)=0$ for any orthogonal $|\Psi)$ and $|\Psi')$. Thus, $|H_{t_{1}t_{2}t_{3}t_{4}})=\sum_i \gamma_i|e_i) \odot |\Psi)$.

It is important
to note that these considerations are related to $|\Psi)(\Psi|$ - observable and the particular history $|\Psi)$. Yet, other histories in the Mach-Zehnder interferometer are also accessible. It shows clearly a physical sense
of quantum entanglement in time and further a concept of its monogamy for a particular entangled history.

Therefore, basing on the above observations, we find temporal monogamy phenomenon for a particular entangled history of similar nature to the spatial monogamy of quantum states \cite{Wootters}. On the ground of consistent histories approach, it implies that we cannot build a tripartite (i.e. defined on three different times) history state $\rho_{t_{3}t_{2}t_{1}}$ where
$\rho_{t_{3}t_{2}}=\rho_{t_{2}t_{1}}=|\Psi )( \Psi|$ and $Tr_{t_{1}}\rho_{t_{3}t_{2}t_{1}}=\rho_{t_{3}t_{2}}$.

\begin{figure}[h]
\centerline{\includegraphics[width=10cm]{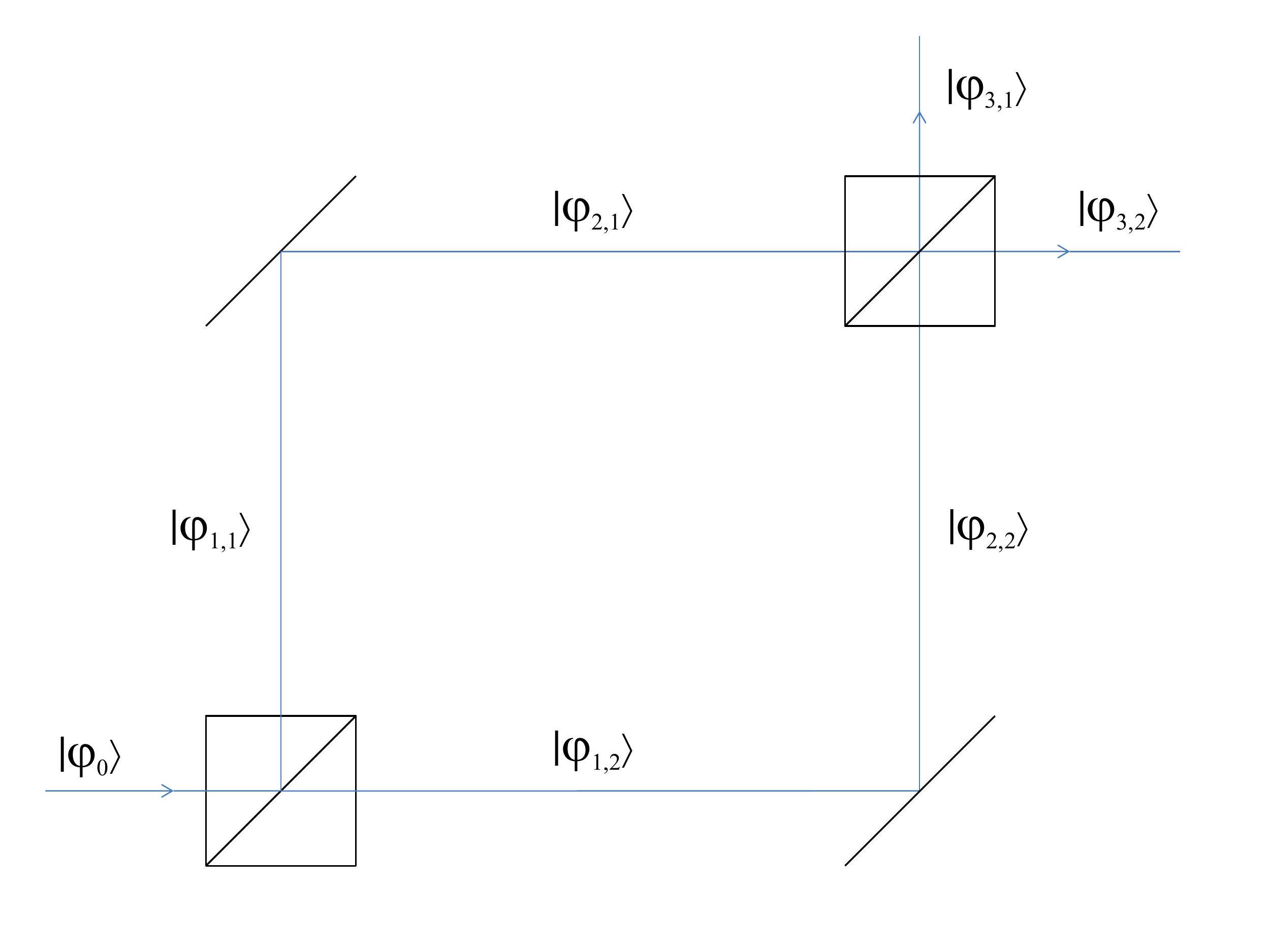}}
 \caption{The Mach-Zehnder interferometer with an input state $|\phi_{0}\rangle$ - a vacuum state is omitted which does not change further considerations. The beam-splitters can be represented by Hadamard operation
 acting on the spatial modes. One can analyze the interferometer via four-times histories on times $t_{0}<t_{1}<t_{2}<t_{3}$ for the interferometer process: $|\phi_{0}\rangle \rightarrow \frac{1}{\sqrt{2}}(|\phi_{1,1}\rangle+|\phi_{1,2}\rangle)\rightarrow \frac{1}{\sqrt{2}}(|\phi_{2,1}\rangle+|\phi_{2,2}\rangle)\rightarrow |\phi_{3,2}\rangle$.}
\end{figure}

Besides the aforementioned reasoning derived from Feynman's quantum paths, one can refer to a broadly used explanation \cite{Wootters} for spatial monogamy of entanglement
between parties ABC (or further $\{t_{3}, t_{2}, t_{1}\}$ for temporal correlations). As mentioned in previous section, it states that A cannot be simultaneously fully entangled with B and C since then AB would be entangled with C having a mixed density
matrix that contradicts purity of the maximally entangled state shared between A and B.

For the history spaces one can build naturally $\mathcal{C}^{*}$-Algebra of history operators equipped with a partial trace operation (\ref{PartialTrace}) and follow the same reasoning for entangled histories.
We can summary these considerations with the following corollary about monogamy of temporal entangled histories \cite{MNPH5}:\\
\\
{\bf Corollary 1.}\\
\textit{
There does not exist  any such a history $|H)\in Proj(\mathcal{H}^{\otimes n})$ so that for three chosen times $\{t_{3},t_{2},t_{1}\}$ one can find reduced histories $|\Psi_{t_{3}t_{2}})=\frac{1}{\sqrt{2}}(|e_{0})\odot|e_{0})+|e_{1})\odot|e_{1}))$ and $|\Psi_{t_{2}t_{1}})=\frac{1}{\sqrt{2}}(|e_{0})\odot|e_{0})+|e_{1})\odot|e_{1}))$.
}
\\

%It is fundamental to note that both histories $|\Psi_{t_{3}t_{2}})$ and $|\Psi_{t_{2}t_{1}})$ inherits the bridging operators %$\{\mathcal{B}(t_{1},t_{2}),\mathcal{B}(t_{2},t_{3})\}$ from the history $|H)$.
%To be in agreement with consistent histories condition one can always extend the history $|\Psi_{t_{3}t_{2}})$ to $|H^{\alpha})=|\Psi_{t_{3}t_{2}})\odot id_{t_{1}}$ and $|\Psi_{t_{2}t_{1}})$ to
%$|H^{\beta})=id_{t_{3}}\odot |\Psi_{t_{2}t_{1}})$ keeping the same bridging operators for both histories as defined for $|H)$, $id_{t_{x}}$ states for identity operator representing actually no
%knowledge about the given instance of time.

This lemma holds for any finite dimension $n$ and also for general entangled states of the form (\ref{maxent}).

%As a consequence, there does not exist such a temporal observable %$\widehat{\Lambda}_{A_{1}A_{2}A_{3}}$ so that $A_{1}A_{2}$ parties are maximally %entangled and $A_{2}A_{3}$ are maximally entangled simultaneously on times %$\{t_{3}, t_{2}, t_{1}\}$. However, in principle there exist
%observables of different histories that do not commute and cannot be observed at %the same reference frame by an observer that are maximally entangled between %$A_{1}A_{2}$ and $A_{2}A_{3}$ \cite{MNPH3}.
%(That would be an explanation for Legett-Garg temporal inequlities)

%Sharing spatial and temporal correlations? Degrees of freedom??
%What about correlations shared between spatial and temporal objects in the same frame of reference F? Does it mean that if we take a spatial pair $\Psi_{AB}$, %then A-party cannot be quantum correlated with other instances of time?

\section{BOUNDING TEMPORAL CORRELATIONS}

For many years there has been studied the violation of local realism (LR) \cite{Bell} and macrorealism (MR) \cite{MRealism} in relation to quantum theories in experimental setups where measurement outputs are tested against violation of Bell inequalities for LR and Leggett-Garg inequalities (LGI) \cite{LGI} for MR. For quantum theories, the former raises as a consequence of non-classical correlations in space while the latter as a consequence of non-classicality of dynamic
evolution. Macrorealism consists of the following assumptions about the reality:

\textit{Macrorealism. } A physical object is at any 'given' time at a definite quantum state.

\textit{Noninvasive measurability.} It is possible to determine the state of the object without any effect on the state and the subsequent evolution.

\textit{Induction.} The properties of an ensemble of quantum states are determined by the initial conditions exclusively (and not by the final conditions.)

In this section we recall the result \cite{MNPH5} that entangled histories approach gives the same well-known Tsirelson bound \cite{Tsirelson} on quantum correlations for LGI as quantum entangled states in case of bi-partite spatial correlations for CHSH-inequalities which saturates the inequalities by quantum mechanical probability distributions.

We take a temporal version of CHSH inequality which is a modification of Legett-Garg inequalities. Then Alice performs a measurement at time $t_{1}$, choosing between two dichotomic observables $\{A_{1}^{(1)}, A_{2}^{(1)}\}$. Bob performs a measurement at time $t_{2}$ choosing between $\{B_{1}^{(2)}, B_{2}^{(2)}\}$.

Then, for such a scenario the Leggett-Garg inequality can be represented in the following form \cite{Vedral}:
\begin{equation}
S_{LGI}\equiv c_{12}+c_{21}+c_{11}-c_{22} \leq 2
\end{equation}
where $c_{ij}=\langle A_{i}^{(1)}, B_{j}^{(2)}  \rangle$ stands for the expectation value of consecutive measurements performed at time $t_{1}$ and $t_{2}$.

Since history operators build a $\mathcal{C}^{*}$-Algebra for normalized histories from projective Hilbert spaces equipped with a well-defined inner product, one can provide reasoning about bounding the LGI purely on the space of entangled histories, and achieve the quantum bound $2\sqrt 2$ of CHSH-inequality specific for spatial correlations.

The importance of this analytical result is due to the fact that
previously it was derived basing on convex optimization methods by means of semi-definite programming \cite{Budroni} and by means of correlator spaces \cite{Fritz} (related to probability conditional distributions
of consecutive events).

We will now recall the theory by B.S. Cirel'son about bounds on Bell's inequalities that is broadly used for finding quantum bounds on spatial Bell-inequalities:
\begin{thm}\cite{Tsirelson}\\
$1.$ There exists $\mathcal{C}^{*}$-Algebra $\mathcal{A}$ with identity, Hermitian operators $A_{1},\ldots, A_{m}, B_{1},\ldots, B_{n} \in \mathcal{A}$ and a state $f$ on $\mathcal{A}$ so that for every ${k,l}$:\\
 \begin{equation}
 A_{k}B_{l}=B_{l}A_{k}; \; \mathbb{I}\leq A_{k}\leq \mathbb{I}; \; \mathbb{I}\leq B_{l}\leq \mathbb{I}; \; f(A_{k}B_{l})=c_{kl}.
\end{equation}
$2. $ There exists a density matrix $W$ such that for every $k,l$:
 \begin{equation}
Tr(A_{k}B_{l}W)=c_{kl} \; and \; A_{k}^{2}=\mathbb{I};\; B_{l}^{2}=\mathbb{I}.
\end{equation}
$3. $ There are unit vectors $x_{1}, \ldots, x_{m}, y_{1},\ldots, y_{n}$ in a $(m+n)$-dimensional Euclidean space such that:
\begin{equation}
\langle x_{k},y_{l} \rangle = c_{kl}.
\end{equation}
\end{thm}
For a temporal setup one considers measurements $\mathbb{A}=I \odot \mathbb{A}^{(1)}$ (measurement $\mathbb{A}$ occurring at time $t_{1}$) and $\mathbb{B}=\mathbb{B}^{(2)}\odot I$, which will in an exact analogy to the proof of the above theorem for a spatial setup \cite{MNPH5}.

The history with 'injected' measurements could be represented as $|\widetilde{H})=\alpha \mathbb{A}\mathbb{B}|H)\mathbb{A}^{\dagger}\mathbb{B}^{\dagger}$ where $\alpha$ stands for a normalization factor.
The history observables are history state operators which are Hermitian and their eigenvectors can generate a consistent history family\cite{WC1}.

For an exemplary observable $A=\sum_{i}a_{i}|H_{i})(H_{i}|$, its measurement on a history $|H)$ generates an expectation value $\langle A\rangle=Tr(A|H)(H|)$ (i.e. the result $a_{i}$ is achieved with probability $|(H|H_{i})|^{2}$) in analogy to the spatial case.

Thus, one achieves a history $|\widetilde{H})$ as a realized
history with measurements and the expectation value of the history observable $\langle A \rangle$.

It is noticeable that $|\widetilde{H})$ and $|H)$ are both compatible histories, i.e. related by a linear transformation. Thus basing on these results, we can state the following lemma:
\begin{lem}
For any history density matrix $W$ and Hermitian history dichotomic observables $A_{i}=I\odot A_{i}^{(1)}$ and $B_{j}= B_{j}^{(2)} \odot I$ where $i,j \in \{1,2\}$ the following bound holds:
\begin{eqnarray}
S_{LGI}&=&c_{11}+c_{12}+c_{21}-c_{22}\\
&=&Tr((A_{1}B_{1}+A_{1}B_{2}+A_{2}B_{1}-A_{2}B_{2})W)\nonumber \\
&\leq& 2\sqrt{2}
\nonumber
\end{eqnarray}
\end{lem}
\begin{proof}
The proof of this observation can be performed in similarity to the spatial version of CHSH-Bell inequality under assumption that the states are represented by
entangled history states and for two possible measurements $\{A_{1}^{(1)}, A_{2}^{(1)}\}$ at time $t_{1}$ and two measurements $\{B_{1}^{(1)}, B_{2}^{(1)}\}$ at time $t_{2}$. These operators can be of dimension $2\times2$ meeting the condition $A_{i}^{2}=B_{j}^{2}=I$. Therefore, they can be interpreted as spin components along two different directions. In consequence, it is well-known that the above inequality is saturated for $2\sqrt{2}$ for a linear combination of tensor spin correlation that holds also for temporal correlations.
Additionally, one could also apply for this temporal inequality reasoning based on the following obvious finding \cite{Tsirelson} that holds also for the temporal scenario due to the structure of $\mathcal{C}^{*}$-Algebra of history operators with $\odot$-tensor operation:
\begin{eqnarray}
A_{1}B_{1}+A_{1}B_{2}+A_{2}B_{1}-A_{2}B_{2}&\leq&\\
\frac{1}{\sqrt{2}}(A_{1}^2+A_{2}^2+B_{1}^2+B_{2}^2)&\leq&
2\sqrt{2}I. \nonumber
\end{eqnarray}
\end{proof}

\section{CONCLUSIONS}
In this paper we presented a concept of quantum entanglement in time on the ground of the consistent histories framework in the extended version including entangled histories.
We introduced a necessary partial trace operator over histories which simplify analysis of reduced histories. Moreover, we discussed
monogamous properties of a particular quantum entangled history proving that
quantum entanglement in time has properties similar to quantum entanglement in space.

It has been also pointed out that
a Tsirelson-like bound can be calculated for Leggett-Garg inequalities analytically applying entangled histories which
is a new result in comparison to the limits calculated numerically by means of semi-definite programming.

However, there are still many open problems and questions related to this field. Entangled histories approach is a substantial modification of the original consistent histories approach, especially in relation to the entanglement in time introducing non-locality of time into the theory.
Future research can be focused on analysis of non-locality in time and finding more appropriate mathematical structures that will enable easier analysis of measurements in different reference frames.

Monogamy of entanglement in time and non-locality in time can be most likely applied to quantum cryptography, and should give new insights into non-sequential quantum algorithms and information processing. Finally, this matter is fundamental for understanding relativistic quantum information theory and brings new prospects for quantum gravity theory.

% Sections that will go in second font

% Acknowledgement
\section{ACKNOWLEDGMENTS}
Acknowledgments to Pawel Horodecki for critical comments and discussions related to this paper. This work was supported partially by the ERC
grant QOLAPS. Part of this work was performed at the National Quantum Information Center in Gdansk.

% References

\nocite{*}
\bibliographystyle{aipnum-cp}%
\bibliography{sample}%

\end{document}